\documentclass[12pt,a4paper,reqno]{amsart}
\usepackage{amsmath, amsfonts,
   amssymb, amsbsy,
   amsthm, euscript}

\newtheorem{theor}{Theorem}
\theoremstyle{definition}

\newtheorem{state}[theor]{Proposition}

\newtheorem{cor}[theor]{Corollary}

\newtheorem{define}{Definition}

\newtheorem{example}{Example}

\theoremstyle{remark}
\newtheorem{rem}{Remark}

\newcommand{\cEv}{\partial} 

\newcommand{\BBR}{\mathbb{R}}

\newcommand{\BBZ}{\mathbb{Z}}

\newcommand{\cA}{\mathcal{A}}

\newcommand{\cE}{\mathcal{E}}

\newcommand{\cELiou}{{\cE}_{\text{\textup{Liou}}}}
\newcommand{\cEToda}{{\cE}_{\text{\textup{Toda}}}}

\newcommand{\cH}{\mathcal{H}}

\newcommand{\bc}{{\mathbf{c}}}
\newcommand{\bp}{{\boldsymbol{p}}}
\newcommand{\bq}{{\boldsymbol{q}}}

\newcommand{\bA}{{\boldsymbol{A}}}

\newcommand{\bun}{\mathbf{1}}

\newcommand{\gothg}{\mathfrak{g}}

\newcommand{\vph}{\varphi}

\newcommand{\Id}{{\mathrm d}}

\DeclareMathOperator{\img}{im}

\DeclareMathOperator{\sym}{sym}

\DeclareMathOperator{\End}{End}

\DeclareMathOperator{\CDiff}{\mathcal{C}Diff}

\DeclareMathOperator{\volume}{vol}

\newcommand{\lshad}{[\![}
\newcommand{\rshad}{]\!]}

\newcommand{\by}[1]{\textit{{#1}}}
\newcommand{\jour}[1]{\textit{{#1}}}
\newcommand{\vol}[1]{\textbf{{#1}}}
\newcommand{\book}[1]{\textrm{{#1}}}

\newcommand{\ib}[3]{ \{\!\{ {#1},{#2} \}\!\}_{{#3}} }

\title[Involutive distributions of operator\/-\/valued evolutionary
fields]%
{Involutive distributions of\\ operator\/-\/valued evolutionary vector fields\\ and their affine geometry}

\date{June 22, 2010} 

\author[A. V. Kiselev]{Arthemy V. Kiselev${}^*$}
\thanks{
        \textit{Address}:
Mathematical Institute, University of Utrecht, P.O.Box~80.010, 3508~TA Utrecht, The Netherlands.
\textit{E-mails}: [\texttt{A.V.Kiselev}, \texttt{J.W.vandeLeur}]%
\texttt{\symbol{"40}uu.nl}.}
\thanks{${}^*$\,Corresponding author. \textit{Current address}: 
SISSA, Via Bonomea~265, 34136~Trieste, Italy.}

\author[J. W. van de Leur]{Johan W. van de Leur}

\subjclass[2000]{
17B66, 
37K30, 
58A30; 
   secondary
17B80, 
37K05, 
47A62. 
}

\keywords{Involutive distributions, Lie algebroids, jet bundles, 
   brackets, Christoffel symbols} 

\begin{document}
\begin{abstract}
We generalize the notion of a Lie algebroid over infinite jet bundle by replacing the variational anchor with an $N$-\/tuple of differential operators whose images in the Lie algebra of evolutionary vector fields of the jet space are subject to collective commutation closure. The linear space of such operators 
becomes an algebra with bi\/-\/differential structural constants, 
of which we study the canonical structure. 
In particular, we show that these constants incorporate bi\/-\/differential analogues of Christoffel symbols.
\end{abstract}
\maketitle
\thispagestyle{empty}
\enlargethispage{0.7\baselineskip}

\subsection*{Introduction}
Lie algebroids~\cite{Vaintrob} are an important and convenient construction that appears, e.g., in classical Poisson dynamics~\cite{PoissonGroupoid} or the theory of quantum Poisson manifolds~\cite{AKZS,Voronov}. Essentially, Lie algebroids extend the tangent bundle~$TM$ over a smooth manifold~$M$, retaining the information about the $C^\infty(M)$-\/module structure for its sections. 
In the paper~\cite{Leiden} we defined the Lie algebroids over the infinite jet spaces for mappings between smooth manifolds (e.g., from strings to space\/-\/time); the classical definition~\cite{Vaintrob} is recovered
by shrinking the source manifold to a point. 
A special case of Lie algebroids over spaces of finite jets for sections
of the tangent bundle was first considered in~\cite{Kumpera}.
Within the variational setup, the anchors become linear matrix differential operators 
that map sections which belong to horizontal modules~\cite{Topical} to the generating sections~$\vph$ of evolutionary derivations~$\cEv_\vph$ on the jet space; by assumption, the images of such anchors are closed under commutation in the Lie algebra of evolutionary vector fields. The two main examples of variational anchors are the recursions with involutive images (\cite{d3Bous}) and the Hamiltonian operators (see~\cite{Olver,Opava,Topical} and~\cite{d3Bous}) whose domains consist of the variational vectors and covectors, respectively.

In~\cite{d3Bous} we studied the \textsl{linear compatibility} of variational anchors, meaning that $N$ 
operators with the common domain span the $N$-\/dimensional linear space~$\cA$ such that each point~$A_{\boldsymbol{\lambda}}\in\cA$ is itself an anchor with involutive image. For example, 
Poisson compatible 
Hamiltonian operators are linear compatible and \textit{vice versa}
(Hamiltonian operators are \textsl{Poisson compatible} if
their linear combinations remain Hamiltonian).
The linear compatibility\footnote{When the set of admissible linear combinations 
$\{\boldsymbol{\lambda}\}\subsetneq\BBR^N$ has punctures near which the homomorphisms~$A_{\boldsymbol{\lambda}}$ exhibit a nontrivial analytic behaviour, this 
concept 
reappears in the theory of continuous contractions of Lie algebras (see~\cite{Contractions} and references therein).}
allows us to reduce the case of many operators~$A_1$,\ $\ldots$, $A_N$ to one operator~$A_{\boldsymbol{\lambda}}=\sum\lambda_i\cdot A_i$ with the same properties.

In this paper 
we introduce a different notion of compatibility for the~$N$ operators. Strictly speaking, we consider the class of structures which is wider than the set of Lie algebroids over the jet spaces. Namely, we relax the assumption that each operator alone is a variational anchor, but, instead, we deal with $N$-tuples of total differential operators~$A_1$,\ $\ldots$,\ $A_N$ whose images are subject to the collective commutation closure:
\[
\Bigl[\sum_{i=1}^N\img A_i,\sum_{j=1}^N\img A_j\Bigr]\subseteq
\sum_{k=1}^N\img A_k.
\]
This involutivity condition converts the linear space 
of operators to an algebra with bi\/-\/differential structural constants~$\bc^k_{ij}$, see~\eqref{DefC} below.
The Magri scheme~\cite{Magri} for the restrictions of compatible Hamiltonian operators to the hierarchy of Hamiltonians yields an example of such overlapping for~$N=2$ with $\bc^k_{ij}\equiv0$.

We study the standard decomposition of the structural constants~$\bc_{ij}^k$, which is similar to the previously known 
case~\eqref{EqDefFrob} for~$N=1$ 
(\cite{SymToda,d3Bous,Leiden}). 
From the bi\/-\/differential constants~$\bc^k_{ij}$ 
we extract the components~$\Gamma^k_{ij}$
that act by total differential operators on both arguments 
at once. 
Our main result, 
Theorem~\ref{GammaIsChristoffel}, states that, under 
a change of coordinates in the domain, the symbols~$\Gamma^k_{ij}$ are
transformed by a proper analogue~\eqref{TransformGamma}
of the classical rule $\Gamma\mapsto g\,\Gamma g^{-1}+\Id g\,g^{-1}$ for
the connection $1$-\/forms $\Gamma$ and reparametrizations~$g$.
We note 
that the bi\/-\/differential symbols~$\Gamma_{ij}^k$
are symmetric in lower indices if the common domain of the~$N$ operators~$A_i$ consists of the variational covectors and hence its elements acquire 
their own odd grading.\footnote{Throughout this paper we deal with a purely commutative setup, refraining from the treatment of super\/-\/manifolds. However, we emphasize that, on a super\/-\/manifold, the two notions of \textsl{parity} and \textsl{grading} (or \textsl{weight}) 
may be totally uncorrelated, see~\cite{Voronov}.
}

This note is organized as follows.
In section~\ref{SecStrong} we introduce the operators with 
collective commutation closure. For consistency, we recall here the cohomological 
formulation~\cite{JK1988}
 of the Magri scheme which gives us an example.
In section~\ref{SecGamma} we study the properties of the bi\/-\/differential constants that appear in such algebras of operators. The analogues of Christoffel symbols emerge 
here; as an example, we calculate them for the symmetry algebra of the Liouville equation.

\section{Compatible differential operators}\label{SecStrong}
\noindent%
We begin with some notation; for a more detailed exposition of the 
geometry of integrable systems we refer to~\cite{Olver} 
and~\cite{Dubr,Opava,ClassSym,Manin1979}.
In the sequel, the ground field is the field~$\BBR$ of real numbers and all mappings are $C^\infty$-\/smooth.

Let $\pi\colon E^{m+n}\xrightarrow[M^m]{}B^n$ be a vector bundle over an 
orientable $n$-\/dimensional manifold~$B^n$ and, 
similarly, let $\xi\colon N^{d+n}\xrightarrow[{}\quad{}
]{}B^n$ be another linear bundle over~$B^n$. Consider the bundle $\pi_\infty\colon J^\infty(\pi)\to B^n$ of infinite jets of sections for the bundle~$\pi$ and take the pull\/-\/back $\pi_\infty^*(\xi)\colon N^{d+n}\mathbin{{\times}_{B^n}}J^\infty(\pi)\to J^\infty(\pi)$ of the bundle~$\xi$ along~$\pi_\infty$. By definition, the $C^\infty(J^\infty(\pi))$-\/module of sections $\Gamma\bigl(\pi_\infty^*(\xi)\bigr)=\Gamma(\xi)\mathbin{{\otimes}_{C^\infty(B^n)}}C^\infty(J^\infty(\pi))$ is called~\textsl{horizontal}, see~\cite{Topical} for further details.

For example, let~$\xi\mathrel{{:}{=}}\pi$. Then the variational vectors~$\vph\in\Gamma\bigl(\pi_\infty^*(\pi)\bigr)$ are the generating sections of evolutionary derivations~$\cEv_\vph$ on~$J^\infty(\pi)$. For convenience, we shall use the shorthand notation~$\varkappa(\pi)\equiv\Gamma\bigl(\pi_\infty^*(\pi)\bigr)$ and $\Gamma\Omega\bigl(\xi_\pi\bigr)\equiv\Gamma\bigl(\pi_\infty^*(\xi)\bigr)$.

Let us consider first the case $N=1$ when there is only one total 
differential operator 
$A\colon\Gamma\Omega(\xi_\pi)\to\varkappa(\pi)$ 
with involutive image: 
\begin{equation}\label{EqDefFrob}
[\img A,\img A]\subseteq\img A.
\end{equation}
The operator~$A$ transfers the bracket in the Lie 
algebra~$\gothg(\pi)=\bigl(\varkappa(\pi),[\,,\,]\bigr)$
to the Lie algebra structure $[\,,\,]_A$ on the quotient of
its domain by the kernel. The standard decomposition of this bracket
is~\cite{d3Bous,Leiden}
\begin{equation}\label{EqOplusBKoszul} 
[\bp,\bq]_A=\cEv_{A(\bp)}(\bq)-\cEv_{A(\bq)}(\bp)+ \{\!\{\bp,\bq\}\!\}_A.
\qquad \bp,\bq\in\Gamma\Omega(\xi_\pi).
\end{equation}
The \textsl{linear compatibility} of operators~\eqref{Ops}, 
which means that their arbitrary linear 
combinations~$A_{\boldsymbol{\lambda}}=\sum_i\lambda_i\cdot A_i$
satisfy~\eqref{EqDefFrob}, reduces the case of~$N\geq2$ operators to the previous case with~$N=1$ as follows.

\begin{theor}[\cite{d3Bous}]\label{ThLinear}
The bracket $\ib{\,}{\,}{A_{\boldsymbol{\lambda}}}$ induced by
the combination $A_{\boldsymbol{\lambda}}=\sum_i\lambda_i\cdot A_i$
on the domain of linear compatible normal{\,}\footnote{By definition,
a total differential operator~$A$ is \textsl{normal} if
$A\circ\nabla=0$ implies $\nabla=0$; in other words, it may be that
$\ker A\neq0$, but the kernel does not have any functional freedom for its elements, see~\cite{SymToda}.}
operators $A_i$ is
\[
\ib{\,}{\,}{\sum\limits_{i=1}^N\lambda_i A_i}=
 \sum_{i=1}^N\lambda_i\cdot\ib{\,}{\,}{A_i}.
\]
The pairwise linear compatibility implies the collective linear compatibility
of $A_1,\ldots,A_N$. 
\end{theor}

\begin{proof}
Consider the commutator $\bigl[\sum_i\lambda_iA_i(\bp),\sum_j\lambda_j A_j(\bq)\bigr]$, here $\bp,\bq\in\Gamma\Omega(\xi_\pi)$.
On one hand, it is equal to
\begin{align}
{}&=\sum_{i\neq j}\lambda_i\lambda_j\bigl[A_i(\bp),A_j(\bq)\bigr]+
  \sum_i\lambda_i^2 A_i\bigl(\cEv_{A_i(\bp)}(\bq)-\cEv_{A_i(\bq)}(\bp)+\ib{\bp}{\bq}{A_i}\bigr).
\label{SqLambdaApart}\\
\intertext{On the other hand, the linear compatibility of~$A_i$ implies}
{}&=A_{\boldsymbol{\lambda}} 
\bigl(\cEv_{A_{\boldsymbol{\lambda}} 
(\bp)}(\bq)\bigr) - A_{\boldsymbol{\lambda}}
\bigl(\cEv_{A_{\boldsymbol{\lambda}}   
(\bq)}(\bp)\bigr)  +  A_{\boldsymbol{\lambda}}   
\bigl(\ib{\bp}{\bq}{A_{\boldsymbol{\lambda}}   
}\bigr).\notag
\end{align}
The entire commutator is quadratic homogeneous in~$\boldsymbol{\lambda}$, whence the
bracket $\ib{\,}{\,}{A_{\boldsymbol{\lambda}}}$ is linear 
in~$\boldsymbol{\lambda}$.
From~\eqref{SqLambdaApart} we see that the individual brackets~$\ib{\,}{\,}{A_i}$
are contained in~it. Therefore,
\[
\ib{\bp}{\bq}{A_{\boldsymbol{\lambda}}}=\sum_\ell\lambda_\ell\cdot\ib{\bp}{\bq}{A_\ell}+
  \sum_\ell\lambda_\ell\cdot\gamma_\ell(\bp,\bq),
\]
where $\gamma_\ell\colon\Gamma\Omega(\xi_\pi)
\times\Gamma\Omega(\xi_\pi)\to\Gamma\Omega(\xi_\pi)$. 

We claim that all summands $\gamma_\ell(\cdot,\cdot)$,
which do not depend on $\boldsymbol{\lambda}$ at all, vanish. Indeed, assume the converse.
Let there be $\ell\in[1,\ldots,N]$ such that $\gamma_\ell(p,q)\neq0$; without loss
of generality, suppose $\ell=1$. Then set $\boldsymbol{\lambda}=(1,0,\ldots,0)$, whence
\begin{multline*}
\Bigl[\sum_i\lambda_iA_i(\bp),\sum_j\lambda_jA_j(\bq)\Bigr]=
\Bigl[\bigl(\lambda_1A_1\bigr)(\bp),\bigl(\lambda_1A_1\bigr)(\bq)\Bigr]
=\bigl(\lambda_1A_1\bigr)\bigl(\lambda_1\gamma_1(\bp,\bq)\bigr)\\
{}+\bigl(\lambda_1A_1\bigr)\Bigl(\cEv_{(\lambda_1A_1)(\bp)}(\bq)-
   \cEv_{(\lambda_1A_1)(\bq)}(\bp)+ 
    \lambda_1\ib{\bp}{\bq}{A_1} 
\Bigr)=\lambda_1A_1\bigl(\lambda_1[\bp,\bq]_{A_1}\bigr).
\end{multline*}
Consequently, $\gamma_\ell(p,q)\in\ker A_\ell$ for all $\bp$ and~$\bq$.
Now we use the 
assumption that each operator~$A_\ell$ is normal.
This implies that $\gamma_\ell=0$ for all $\ell$, 
which concludes the proof.
\end{proof}

Now we let~$N>1$ and consider $N$-\/tuples of linear 
total differential operators 
\begin{equation}\label{Ops}
A_1,\ldots,A_N\colon\Gamma\Omega\bigl(\xi_\pi\bigr)\longrightarrow\varkappa(\pi),
\end{equation}
whose images in the Lie algebra~$\gothg(\pi)$ of evolutionary vector fields
on~$J^\infty(\pi)$ are subject to collective commutation 
closure. 

\begin{define}
We say that $N\geq2$ total differential operators~\eqref{Ops}
are \textsl{strong compatible}
if the sum of their images is closed under commutation 
in the Lie algebra~$\gothg(\pi)=\bigl(\varkappa(\pi),[\,,\,]\bigr)$ of evolutionary vector fields,
\begin{equation}\label{CommutClosure}
\Bigl[\sum\limits_i\img A_i,\sum\limits_j\img A_j\Bigr]\subseteq
\sum\limits_k\img A_k, \qquad 1\leq i,j,k\leq N.
\end{equation}
The involutivity~\eqref{CommutClosure} gives rise to the
bi\/-\/differential operators~$\bc^k_{ij}\colon
\Gamma\Omega(\xi_\pi)\times\Gamma\Omega(\xi_\pi)\to\Gamma\Omega(\xi_\pi)$ through
\begin{equation}\label{DefC}
\bigl[A_i(\bp),A_j(\bq)\bigr]=\sum_k A_k\bigl(\bc^k_{ij}(\bp,\bq)\bigr),\qquad
 \bp,\bq\in\Gamma\Omega(\xi_\pi).
\end{equation}
The structural constants~$\bc^k_{ij}$ absorb the bi\/-\/differential action 
on~$\bp,\bq$ under the commutation in the images of the operators.
\end{define}

\begin{rem}
If $N=1$ and there is a unique operator~$A\colon\Gamma\Omega(\xi_\pi)\to\varkappa(\pi)$ satisfying~\eqref{EqDefFrob}, then we recover the definition of the variational anchor in the Lie algebroid over the infinite jet space~$J^\infty(\pi)$, see~\cite{Leiden}. 
By construction, $\bc^1_{11}\equiv [\,,\,]_{A_1}$ if $N=1$.
However, for $N>1$ we obtain a wider class of structures because 
we do not assume that the image of each operator~$A_i$ alone is involutive, 
therefore it may well occur that~$\bc^k_{ii}\neq0$ for some~$k\neq i$.
\end{rem}

The Magri scheme~\cite{Magri} for the restriction of two compatible
Hamiltonian operators $A_1,A_2$ 
onto the commutative hierarchy of the descendants 
$\cH_i$ 
of the Casimirs $\cH_0$ for~$A_1$ 
gives us an example of~\eqref{CommutClosure}
with~$N=2$ and~$\bc^k_{ij}\equiv0$.
Let us consider it in more detail;
from now on, we standardly identify the Hamiltonian operators~$A$ with the variational Poisson bi\/-\/vectors~$\bA$, see~\cite{Topical}. 
We recall that the variational Schouten bracket~$\lshad\,,\,\rshad$
of such 
bi\/-\/vectors satisfies the 
Jacobi identity
\begin{equation}\label{Jacobi4Schouten}
\lshad\lshad\bA_1,\bA_2\rshad,\bA_3\rshad+
\lshad\lshad\bA_2,\bA_3\rshad,\bA_1\rshad+
\lshad\lshad\bA_3,\bA_1\rshad,\bA_2\rshad=0.
\end{equation}
Hence the defining property 
$\lshad\bA,\bA\rshad 
=0$ for a Poisson bi\/-\/vector
~$\bA$ implies that $\Id_{A}=\lshad\bA,{\cdot}\,\rshad$ is a
differential, giving rise to the Poisson cohomology $H_A^k$. 
Obviously, the Casimirs $\cH_0$ 
such that $\lshad\bA,\cH_0
\rshad=0$ for a Poisson bi\/-\/vector
~$\bA$ constitute the group~$H_A^0$. 

\begin{theor}[
\cite{JK1988,Magri}]\label{ThMagri} 
Suppose $\lshad\bA_1,\bA_2\rshad=0$, $\cH_0\in H_{A_1}^0$ 
is a Casimir of~$\bA_1$, and the first Poisson cohomology w.r.t.
$\Id_{A_1}=\lshad\bA_1,{\cdot}\,\rshad$ vanishes. 
Then for any $k>0$ there is a Hamiltonian $\cH_k$ 
such that
\begin{equation}\label{MagriResolve}
\lshad\bA_2,\cH_{k-1}\rshad=\lshad\bA_1,\cH_k\rshad.
\end{equation}
Put $\vph_k\mathrel{{:}{=}}A_1\bigl(\delta/\delta u
(\cH_k)\bigr)$ such that $\cEv_{\vph_k}=\lshad\bA_1,\cH_k\rshad$.
The Hamiltonians $\cH_i$, $i\geq0$, pairwise Poisson commute w.r.t.\ either
$A_1$ or $A_2$, the densities of~$\cH_i$ are conserved on any
equation $u_{t_k}=\vph_k$, and the evolutionary derivations
$\cEv_{\vph_k}$ pairwise commute for all~$k\geq0$.
\end{theor}

\begin{proof}[Standard proof of existence]
The main homological equality~\eqref{MagriResolve} is established by
induction on~$k$. Starting with a Casimir~$\cH_0$, we obtain
\[
0=\lshad\bA_2,0\rshad=\lshad\bA_2,\lshad\bA_1,\cH_0\rshad\rshad =
 -\lshad\bA_1,\lshad\bA_2,\cH_0\rshad\rshad\mod\lshad\bA_1,\bA_2\rshad=0,
\]
using the Jacobi identity~\eqref{Jacobi4Schouten}.
The first Poisson cohomology $H^1_{A_1}
=0$ is
trivial by an assumption of the theorem, hence the closed element
$\lshad\bA_2,\cH_0\rshad$ in the kernel of $\lshad\bA_1,\cdot\rshad$
is exact: $\lshad\bA_2,\cH_0\rshad=\lshad\bA_1,\cH_1\rshad$
for some~$\cH_1$.
For~$k\geq1$, we have
\[
\lshad\bA_1,\lshad\bA_2,\cH_k\rshad\rshad =
-\lshad\bA_2,\lshad\bA_1,\cH_k\rshad\rshad =
-\lshad\bA_2,\lshad\bA_2,\cH_{k-1}\rshad\rshad =0
\]
using
~\eqref{Jacobi4Schouten}
and by $\lshad\bA_2,\bA_2\rshad=0$.
Consequently, by~$H^1_{A_1}
=0$ we have that $\lshad\bA_2,\cH_k\rshad=\lshad\bA_1,\cH_{k+1}\rshad$, 
and we thus proceed infinitely.
\end{proof}


We see now that the inductive step --- the existence of the $(k+1)$-st Hamiltonian functional in involution~--- is possible if and only if $H_0$ is a Casimir,\footnote{The 
Magri scheme starts from any two Hamiltonians
$\cH_{k-1},\cH_k
$ that satisfy~\eqref{MagriResolve},
but we 
operate with maximal subspaces of the space of functionals
such that the sequence $\{\cH_k\}$ 
can not be extended with~$k<0$.}
and therefore the operators $A_1$~and $A_2$~are restricted onto the linear subspace which is spanned in the space of variational covectors
by the Euler derivatives of the descendants of~$\cH_0$, 
i.e., of the Hamiltonians of the hierarchy. 
We note that 
the image under~$A_2$ of a generic section from
the domain of operators~$A_1$ and~$A_2$
can not be resolved w.r.t.\ $A_1$ by~\eqref{MagriResolve}.
For example, the first and second Hamiltonian structures for the KdV equation,  
which equal, respectively, $A_1=\Id/\Id x$ and 
$A_2=-\tfrac{1}{2}\tfrac{\Id^3}{\Id x^3}+2u\tfrac{\Id}{\Id x}+u_x$,
are not strong
compatible unless they are restricted onto some subspaces of their
arguments. On the linear subspace of descendants of the Casimir 
$\int u\,\Id x$,
we have $\text{im}\,A_2\subset\text{im}\,A_1$ and, since the
image of the Hamiltonian operator~$A_1$ is involutive, we conclude
that $[\text{im}\,A_1$,\ $\text{im}\,A_2]\subset\text{im}\,A_1$.

On the other hand, the strong compatibility of the restrictions of Poisson compatible operators~$A_1$ and~$A_2$ onto the hierarchy is valid 
since their images are commutative Lie algebras. 
Regarding the converse statement as a potential generator of multi\/-\/dimensional completely integrable systems, we formulate the open problem:
Is the strong compatibility of Poisson 
compatible Hamiltonian operators achieved 
\textsl{only} for their restrictions onto the hierarchies of Hamiltonians in involution so that the bi\/-\/differential constants~$\bc^k_{ij}$ 
necessarily vanish\,?

\section{Bi\/-\/differential Christoffel symbols}\label{SecGamma}
\noindent
Similarly to~\eqref{EqOplusBKoszul}, 
we extract the total bi\/-\/differential parts 
of the structural constants~$\bc^k_{ij}$ in~\eqref{DefC}
and obtain
\begin{equation}\label{DefGamma}
\bc^k_{ij}=\cEv_{A_i(\bp)}(\bq)\cdot\delta^k_j
 -\cEv_{A_j(\bq)}(\bp)\cdot\delta^k_i+\Gamma^k_{ij}(\bp,\bq),\qquad 
\bp,\bq\in\Gamma\Omega(\xi_\pi),
\end{equation}
where $\Gamma^k_{ij}\in\CDiff\bigl(\Gamma\Omega(\xi_\pi)\times\Gamma\Omega(\xi_\pi)\to
\Gamma\Omega(\xi_\pi)\bigr)$
and $\delta^k_i$,\ $\delta^k_j$ are the Kronecker delta symbols.
By definition, the three indices in~$\Gamma^k_{ij}$ match the
respective operators~$A_i,A_j,A_k$ in~\eqref{DefC}.
(The total number of the indices is much greater than~three; 
moreover, the proper upper or lower location of the omitted indices
depends on the (co)vector nature of the domain~$\Gamma\Omega(\xi_\pi)$.)
Obviously, the convention 
\[\Gamma^1_{11}=\ib{\,}{\,}{A_1}\]
holds if~$N=1$.
At the same time, for fixed $i,j,k$, the symbol~$\Gamma^k_{ij}$ remains
a (class of) matrix differential operator in each of its two arguments
$\bp,\bq\in\Gamma\Omega(\xi_\pi)$. 
The symbol $\Gamma_{ij}^k$ represents a class of bi\/-\/differential operators
because they are not uniquely defined. Indeed, they are gauged by the conditions
\begin{equation}\label{NonUnique}   
\sum_{k=1}^N A_k\Bigl(\cEv_{A_i(\bp)}(\bq)\delta^k_j
 -\cEv_{A_j(\bq)}(\bp)\delta^k_i
 +\Gamma^k_{ij}(\bp,\bq)\Bigr)=0,\qquad 
\bp,\bq\in\Gamma\Omega(\xi_\pi).
\end{equation}
We let the r.h.s.\ of~\eqref{NonUnique} be zero if
the sum $\sum_\ell\img A_\ell$ of the images is indecomposable,
meaning that no nontrivial sections commute with all the others:
$\bigl[A_k(\bp),\sum_{\ell=1}^N\img A_\ell
\bigr]=0$ implies that $\bp\in\ker A_k$. For this it is sufficient that
the sum of the images of~$A_\ell$ in $\gothg(\pi)$~is semi\/-\/simple
and the Whitehead lemma holds for it~\cite{Fuchs}.
Otherwise, the right\/-\/hand side of~\eqref{NonUnique} belongs to the
linear subspace of such nontrivial sections.

\begin{example}[see~\cite{Protaras,Leiden}]\label{ExStrongFromLiou}
Consider the Liouville equation $\cELiou=\{u_{xy}=\exp(2u)\}$.
The differential generators of its conservation laws
are 
$w=u_x^2-u_{xx}\in\ker\frac{\Id}{\Id y}{\bigr|}_{\cELiou}$ 
{and} 
$\bar{w}=u_y^2-u_{yy}\in\ker\frac{\Id}{\Id x}{\bigr|}_{\cELiou}$.
%
The operators\footnote{We denote the operators by $\square$ 
and~$\overline{\square}$, following the notation of~
\cite{SymToda,Protaras}, see also references therein.} 
$\square=u_x+\tfrac{1}{2}\tfrac{\Id}{\Id x}$ and 
$\overline{\square}=u_y+\tfrac{1}{2}\tfrac{\Id}{\Id y}$
determine higher symmetries $\vph,\bar{\vph}$ 
of~$\cELiou$ by the formulas 
\[
\vph=\square\bigl(p(x,[w])\bigr),\qquad
\bar{\vph}=\overline{\square}\bigl(\overline{p}(y,[\bar{w}])\bigr)
\]
for any variational covectors 
$p,\overline{p}$.
The images of $\square$ and $\overline{\square}$ 
are closed w.r.t.\ the commutation; for instance,
the bracket~\eqref{EqOplusBKoszul} for~$\square$ contains
$\ib{p}{q}{\square}=\tfrac{\Id}{\Id x}(p)\cdot q-p\cdot 
\tfrac{\Id}{\Id x}(q)$,
and similar for $\overline{\square}$. The two summands in
the symmetry algebra $\sym\cELiou\simeq\img\square\oplus\img\overline{\square}$ 
commute between each other,
$  
[\img\square,\img\overline{\square}]\doteq0
$ on~$\cELiou$. 
The operators $\square$, $\overline{\square}$   
generate the bi\/-\/differential symbols
\begin{align*}
\Gamma_{\square\square}^{\square}&=\ib{\,}{\,}{\square}=
\tfrac{\Id}{\Id x}\otimes\bun-\bun\otimes 
\tfrac{\Id}{\Id x},&
\Gamma_{\overline{\square}\,\overline{\square}}^{\overline{\square}}&=
\ib{\,}{\,}{\overline{\square}}=
\tfrac{\Id}{\Id y}\otimes\bun-\bun\otimes \tfrac{\Id}{\Id y},\\ 
\Gamma_{{\square}\overline{\square}}^{{\square}}&=\tfrac{\Id}{\Id y}
  \otimes\bun,\qquad
\Gamma_{{\square}\overline{\square}}^{\overline{\square}}=-\bun\otimes 
\tfrac{\Id}{\Id x},&
\Gamma_{\overline{\square}{\square}}^{{\square}}&=-\bun\otimes 
\tfrac{\Id}{\Id y},\qquad
\Gamma_{\overline{\square}{\square}}^{\overline{\square}}=\tfrac{\Id}{\Id x}
\otimes\bun,
\end{align*}
where the notation is obvious.
We note that $\Gamma_{{\square}\overline{\square}}^{{\square}}(p,q)\doteq
\Gamma_{{\square}\overline{\square}}^{\overline{\square}}(p,q)\doteq
\Gamma_{\overline{\square}{\square}}^{{\square}}(q,p)\doteq
\Gamma_{\overline{\square}{\square}}^{\overline{\square}}(q,p)\doteq0$
on~$\cELiou$ for any $p(x,[w])$ and~$q(y,[\bar{w}])$.

The matrix operators $\square$,\ $\overline{\square}$ 
are well\/-\/defined~\cite{SymToda} for each 2D~Toda  
chain~$\cEToda$ associated with a semi\/-\/simple complex Lie algebra. They exhibit the same properties as above.
\end{example}

\begin{rem}
The operators~$\square$,\ $\overline{\square}$ yield the involutive distributions
of evolutionary vector fields that are tangent to the \textsl{integral
manifolds}, the 2D~Toda differential equations. Generally,
there is no Frobenius theorem for such distributions.
Still, if the integral manifold exists and is an infinite prolongation of a differential equation~$\cE\subset J^\infty(\pi)$,
then, by construction, this equation admits infinitely many symmetries of the form $\vph=A_i(\bp)$ with free functional parameters $\bp\in\Gamma\Omega(\xi_\pi)$. This property is close but not equivalent to the definition of 
Liouville\/-\/type systems 
(see~\cite{SymToda,Protaras} and references therein).
\end{rem}

The method by which we introduced the symbols~$\Gamma_{ij}^k$ suggests that, under re\-pa\-ra\-met\-ri\-za\-ti\-ons~$g$ in the domain of the operators~\eqref{Ops}, they obey a proper 
analogue of the standard rule
$\Gamma\mapsto g \,\Gamma\, g ^{-1}+\Id g \cdot g ^{-1}$ 
for the connection $1$-\/forms~$\Gamma$.
This is indeed so.

\begin{theor}[Transformations of~$\Gamma_{ij}^k$]\label{GammaIsChristoffel}
Let~$g$ be a reparametrization
$\bp\mapsto\tilde{\bp}= g\bp$, $\bq\mapsto\tilde{\bq}= g\bq$
of sections $\bp,\bq\in\Gamma\Omega(\xi_\pi)$
in the domains\footnote{
Under an invertible 
change $\tilde{w}=\tilde{w}[w]$ of fibre coordinates 
(see Example~\ref{ExStrongFromLiou}), 
the variational covectors are transformed by the inverse of the adjoint linearization 
$g =\bigl[\bigl(\ell_{\tilde{w}}^{(w)}\bigr)^{\dagger}\bigr]^{-1}$,
whereas for va\-ri\-a\-ti\-o\-nal vectors, $g =\ell_{\tilde{w}}^{(w)}$
is the linearization.} 
of strong compatible operators~\eqref{Ops}.    
In this notation, the operators $A_1$,\ $\ldots$,\
$A_N$ 
are transformed by the formula 
$A_i\mapsto\tilde{A}_i=A_i\circ g ^{-1}\bigr|_{w=w[\tilde{w}]}$. 
Then the bi\/-\/differential symbols
$\Gamma_{ij}^k\in\CDiff\bigl(\Gamma\Omega(\xi_\pi)\times\Gamma\Omega(\xi_\pi)\to
\Gamma\Omega(\xi_\pi)\bigr)$ 
are transformed according to the rule
\begin{equation}\label{TransformGamma}
\Gamma_{ij}^k(\bp,\bq)\mapsto
{\tilde{\Gamma}}_{\tilde{\imath}\tilde{\jmath}}^{\tilde{k}}
 \bigl(\tilde{\bp},\tilde{\bq}\bigr)=
\bigl( g \circ\Gamma_{\tilde{\imath}\tilde{\jmath}}^{\tilde{k}}\bigr)
  \bigl( g ^{-1}\tilde{\bp}, g ^{-1}\tilde{\bq}\bigr)+
\delta_{\tilde{\imath}}^{\tilde{k}}\cdot\cEv_{\tilde{A}_{\tilde{\jmath}}
  (\tilde{\bq})}( g )\bigl( g ^{-1}\tilde{\bp}\bigr)-
\delta_{\tilde{\jmath}}^{\tilde{k}}\cdot\cEv_{\tilde{A}_{\tilde{\imath}}
  (\tilde{\bp})}( g )\bigl( g ^{-1}\tilde{\bq}\bigr).
\end{equation}
\end{theor}

\begin{proof}
Denote $A=A_i$ and $B=A_j$; 
without loss of generality we assume $i=1$ and~$j=2$.
Let us calculate the commutators of vector fields in
the images of $A$ and~$B$ using two systems of coordinates in the domain. 
We equate the commutators straighforwardly, because the fibre coordinates in
the images of the operators are not touched at all. So, we have, originally,
\begin{align*}
\bigl[A(\bp),&B(\bq)\bigr]=B\bigl(\cEv_{A(\bp)}(\bq)\bigr)-A\bigl(\cEv_{B(\bq)}(\bp)\bigr)+
 A\bigl(\Gamma_{AB}^A(\bp,\bq)\bigr)+B\bigl(\Gamma_{AB}^B(\bp,\bq)\bigr)+
 \sum_{k=3}^N A_k\bigl(\Gamma_{AB}^k(\bp,\bq)\bigr).\\
\intertext{On the other hand, we substitute $\tilde{\bp}= g\bp$ and $\tilde{\bq}= g\bq$
in $\bigl[\tilde{A}(\tilde{\bp}),\tilde{B}(\tilde{\bq})\bigr]$, whence, by the Leibnitz
rule, we obtain}
\bigl[\tilde{A}(\tilde{\bp}),&\tilde{B}(\tilde{\bq})\bigr]=
\tilde{B}\bigl({\cEv_{\tilde{A}(\tilde{\bp})}}   
  ( g )(\bq)\bigr)+\bigl({\tilde{B}\circ g }\bigr)  
  \bigl({\cEv_{\tilde{A}(\tilde{\bp})}}(\bq)\bigr)-   
\tilde{A}\bigl({\cEv_{\tilde{B}(\tilde{\bq})}}( g )(\bp)\bigr)- 
  \bigl({\tilde{A}\circ g }\bigr)   
  \bigl({\cEv_{\tilde{B}(\tilde{\bq})}}(\bp)\bigr)\\  
{}&{}
+\bigl(A\circ g ^{-1}\bigr)\bigl(\Gamma_{\tilde{A}\tilde{B}}^{\tilde{A}}( g  \bp, g  \bq)\bigr)
+\bigl(B\circ g ^{-1}\bigr)\bigl(\Gamma_{\tilde{A}\tilde{B}}^{\tilde{B}}( g  \bp, g  \bq)\bigr)
+\sum_{\tilde{k}=3}^N\bigl(
                           {A}_{\tilde{k}}\circ g ^{-1}\bigr)
  \bigl(\Gamma_{\tilde{A}\tilde{B}}^{\tilde{k}}( g  \bp, g  \bq)\bigr).
\end{align*}
Therefore,
\begin{align*}
\Gamma_{AB}^A(\bp,\bq)&=\bigl( g ^{-1}\circ\Gamma_{\tilde{A}\tilde{B}}^{\tilde{A}}\bigr)
  ( g  \bp, g  \bq) -\bigl( g ^{-1}\circ\cEv_{B(\bq)}( g )\bigr)(\bp),\\
\Gamma_{AB}^B(\bp,\bq)&=\bigl( g ^{-1}\circ\Gamma_{\tilde{A}\tilde{B}}^{\tilde{B}}\bigr)
  ( g  \bp, g  \bq) +\bigl( g ^{-1}\circ\cEv_{A(\bp)}( g )\bigr)(\bq),\\
\Gamma_{AB}^k(\bp,\bq)&=\bigl( g ^{-1}\circ\Gamma_{\tilde{A}\tilde{B}}^{
{k}}\bigr)
  ( g  \bp, g  \bq)\qquad\text{for $k\geq3$.}
\end{align*}
Acting by~$ g $ on these equalities and expressing $\bp= g ^{-1}\tilde{\bp}$,
$\bq= g ^{-1}\tilde{\bq}$, we obtain~\eqref{TransformGamma} and
conclude the proof.
\end{proof}

\begin{rem}\label{RemManin}
Within the Hamiltonian formalism, it is very productive to postulate
that the arguments of Hamiltonian operators, the variational covectors,
are \textsl{odd},\footnote{Here 
we assume for simplicity that all fibre coordinates
in~$\pi$ and~$\xi$   
are permutable.} 
see~\cite{Voronov} and~\cite{Topical}. 
Indeed, in this particular situation they  
can be conveniently identified with Cartan $1$-\/forms times the pull\/-\/back
of the volume form $\Id\volume(B^n)$ for the base of the jet bundle.
We preserve this 
\textsl{grading} 
for such domains 
of 
operators
(when $N=1$, we referred to such operators in~\cite{Leiden}
as variational anchors of \textsl{second kind}). 
If, moreover, $\pi$~and $\xi$~are super\/-\/bundles with Grassmann\/-\/valued
sections, then the operators 
become bi\/-\/graded~\cite{Voronov}.
Their proper 
grading is $-1$ 
because their images in~$\gothg(\pi)$ have grading zero,
but the $\BBZ_2$-\/\textsl{parity}, if any, can be arbitrary.
\end{rem}

\begin{cor}\label{CorSymGamma}
For strong compatible 
operators 
whose domain~$\Gamma\Omega(\xi_\pi)$ consists of variational
covectors, 
the 
grading of the arguments equals~$1$.
Therefore, for any $i,j,k\in[1,\ldots,N]$ and for  
any $\bp,\bq\in\Gamma\Omega(\xi_\pi)$
we have that
\begin{equation}\label{GammaGradedSym}
\Gamma^k_{ij}(\bp,\bq)=-\Gamma_{ji}^k(\bq,\bp)
 =(-1)^{|\bp|_{\text{gr}
}\cdot|\bq|_{\text{gr}
}}\cdot\Gamma^k_{ji}(\bq,\bp)
\end{equation}
due to the 
skew\/-\/symmetry of the commutators in~\eqref{CommutClosure}. 
Hence the symbols $\Gamma^k_{ij}$ are symmetric 
in this case. 
\end{cor}

\begin{state}
If two normal operators $A_i$ and $A_j$ are simultaneously
linear and strong compatible,
then their `individual' 
brackets $\Gamma^i_{ii}$ and $\Gamma^j_{jj}$ are
\[ 
\ib{\bp}{\bq}{A_i}=\Gamma^j_{ij}(\bp,\bq)+\Gamma^j_{ji}(\bp,\bq)
\quad \text{and}\quad
\ib{\bp}{\bq}{A_j}=\Gamma^i_{ij}(\bp,\bq)+\Gamma^i_{ji}(\bp,\bq)
\]
for any~$\bp,\bq\in\Gamma\Omega(\xi_\pi)$.
\end{state}

\begin{proof}
For brevity, denote $A=A_i$, $B=A_j$ and consider the 
linear combination~$\mu A+\nu B$; by assumption, its image is closed under commutation. 
By Theorem~\ref{ThLinear}, we have
\begin{multline*}
\bigl(\mu A+\nu B\bigr)\bigl(\ib{\bp}{\bq}{\mu A+\nu B}\bigr)={}\\
 {}=\mu^2A\bigl(\ib{\bp}{\bq}{A}\bigr)+\mu\nu\cdot A\bigl(\ib{\bp}{\bq}{B}\bigr)
 +\mu\nu\cdot B\bigl(\ib{\bp}{\bq}{A}\bigr)+\nu^2B\bigl(\ib{\bp}{\bq}{A}\bigr).
\end{multline*}
On the other hand,
\begin{multline*}
\Bigl[\bigl(\mu A+\nu B\bigr)(\bp),\bigl(\mu A+\nu B\bigr)(\bq)\Bigr]\\
{}=\mu^2\bigl[A(\bp),A(\bq)\bigr]+\mu\nu\bigl[A(\bp),B(\bq)\bigr]
   -\mu\nu\bigl[A(\bq),B(\bp)\bigr]+\nu^2\bigl[B(\bp),B(\bq)\bigr].
\end{multline*}
Taking into account~\eqref{DefGamma} and equating the coefficients of~$\mu\nu$,
we obtain
\[
A\bigl(\ib{\bp}{\bq}{B}\bigr)+B\bigl(\ib{\bp}{\bq}{A}\bigr)=
 A\bigl(\Gamma_{AB}^A(\bp,\bq)\bigr)+B\bigl(\Gamma_{AB}^B(\bp,\bq)\bigr)-
 A\bigl(\Gamma_{AB}^A(\bq,\bp)\bigr)-B\bigl(\Gamma_{AB}^B(\bq,\bp)\bigr).
\]
Using the 
formulas $\Gamma_{AB}^A(\bq,\bp)=-\Gamma_{BA}^A(\bp,\bq)$ and
$\Gamma_{AB}^B(\bq,\bp)=-\Gamma_{BA}^B(\bp,\bq)$, see~\eqref{GammaGradedSym},
we isolate the arguments of the operators and obtain the assertion.
\end{proof}

\section*{Conclusion}
\noindent%
For every $\Bbbk$-\/vector space $V$, the space of endomorphisms 
$\End_\Bbbk(V)$ is a monoid 
with respect to the composition~$\circ$. In this context, one can study relations between recursion operators.
For instance, the structural relations for
recursion operators of the Krichever\/--\/Novikov equations are described by hyperelliptic curves, see~\cite{SokDemskojKN}.
Likewise, we have the relation $R_1\circ R_2-R_2\circ R_1=R_1^2$ between two recursions for the dispersionless $3$-\/component Boussinesq system, see~\cite{JK3Bous}.
Simultaneously, the space of endomorphisms carries the structure of a Lie algebra, which is given by the formula $[R_i,R_j]=R_i\circ R_j-R_j\circ R_i$ for every $R_i,R_j\in\End_\Bbbk(V)$. 

In this paper we proceed further and consider the class of structures on the linear spaces of total differential operators that, generally, do not in principle admit any associative composition. (The bracket of recursion operators that appears through~\eqref{DefC} is different from the Richardson\/--\/Nijenhuis bracket~\cite{Opava}, although we use similar geometric techniques.)
The classification problem for such algebras of operators is completely open.

\subsection*{Discussion}
We performed all the reasonings for 
local differential operators in a purely commutative setup;
all the structures were defined 
on the empty jet spaces. A rigorous 
extension of these objects to $\BBZ_2$-\/graded 
nonlocal operators on 
differential equations is a separate problem for future research.
In addition, the use of difference 
operators subject to~\eqref{CommutClosure}
can be a fruitful idea \emph{au d\'ebut} 
for discretization of integrable systems with free
functional parameters in the symmetries 
(e.g., Toda\/-\/like 
difference systems~\cite{Suris}). 

\subsection*{Acknowledgements}
This work has been partially supported by the European Union through
the FP6 Marie Curie RTN 
{ENIGMA} (Contract
no.\,MRTN-CT-2004-5652), the European Science Foundation Program
{MISGAM}, and by NWO grants~B61--609 and VENI~639.031.623. 
%
%
A.\,K.\ thanks Max Planck Institute for Mathematics (Bonn) 
and SISSA
for financial support and warm hospitality.


\begin{thebibliography}{99}

\bibitem{AKZS}
\by{Alexandrov M., Schwarz A., Zaboronsky O., Kontsevich M.} (1997)
The geometry of the master equation and topological quantum field
theory, \jour{Int.~J.\ Modern Phys.~A} \vol{12}:7, 1405--1429.

\bibitem{PoissonGroupoid}
\by{Crainic M., Fernandes R. L.} (2004) 
Integrability of Poisson brackets, \jour{J.~Diff.\ Geom.} \vol{66},
71--137.

\bibitem{SokDemskojKN}
\by{Demskoi D. K., Sokolov V. V.} (2008)
On recursion operators for elliptic models, \jour{Nonlinearity}
\vol{21}:6, 1253--1264. 

\bibitem{Dubr}
\by{Dubrovin B. A.} (1996) Geometry of 2D topological field theories,
\jour{Lect.\ Notes in Math.} \vol{1620} Integrable systems and quantum
groups (Montecatini Terme, 1993), Springer, Berlin,
120--348.

\bibitem{Fuchs}
\by{Fuks D. B.} (1986)
\book{Cohomology of infinite\/-\/dimensional Lie algebras}.
Contemp.\ 
Sov.\ 
Math., 
Consultants Bureau,~NY. 

\bibitem{JK3Bous}
\by{Kersten P., Krasil'shchik I., Verbovetsky A.} (2006)
A geometric study of the dispersionless Boussinesq type equation,
\jour{Acta Appl.\ Math.} \vol{90}:1--2, 143--178.

\bibitem{SymToda}
\by{Kiselev A. V., van de Leur J. W.} (2010) Symmetry algebras of Lagrangian
Liouville\/-\/type systems, \jour{Theor.\ Math.\ Phys.},
\vol{162}:3, 149--162.\ \texttt{arXiv:nlin.SI/0902.3624}

\bibitem{d3Bous}
\by{Kiselev A. V., van de Leur J. W.} (2009) A family of second Lie algebra
structures for symmetries of dispersionless Boussinesq system,
\jour{J.~Phys.\ A\textup{:} Math.\ Theor.}, \vol{42}:40, 
404011 (8~p.) \texttt{arXiv:nlin.SI/0903.1214}

\bibitem{Protaras}
\by{Kiselev A. V., van de Leur J. W.} (2009)
A geometric derivation of KdV\/-\/type hierarchies from root systems, in:
Proc.\ 4th Int.\ workshop `Group analysis of differential equations and
integrable systems' (October 26--30, 2008;\ Protaras, Cyprus),
87--106.\ \texttt{arXiv:nlin.SI/0901.4866} 


\bibitem{Leiden}
\by{Kiselev A. V., van de Leur J. W.} (2010) Variational Lie algebroids, 
21~p. \jour{Preprint} \texttt{arXiv:math.DG/1006.4227}

\bibitem{JK1988}
\by{Krasil'shchik I. S.} (1988)
Schouten bracket and canonical algebras.
\jour{Global analysis ---
studies and applications.}~III, 
Lecture Notes in Math.\ \vol{1334}
(Yu.~G.~Borisovich and Yu.~E.~Gliklikh, eds.), Springer, Berlin, 79--110.

\bibitem{Opava}
\by{Krasil'shchik I., Verbovetsky A.} (1998)
\book{Homological methods in equations of mathematical physics}.
Open Education and Sciences, Opava. \texttt{arXiv:math.DG/9808130}

\bibitem{Topical}
\by{Krasil'shchik J., Verbovetsky A.} (2010)
Geometry of jet spaces and integrable systems.
\jour{Preprint} \texttt{arXiv:math.DG/1002.0077}, 63~p.

\bibitem{ClassSym}
\by{Krasil'shchik I. S., Vinogradov A. M.}, eds.\ (1999)
\book{Symmetries and conservation laws for differential equations of mathematical physics}. (
{Bocharov A. V., Chetverikov V. N., Duzhin S. V.  \textit{et al}.})
AMS, Providence, RI.

\bibitem{Kumpera}
\by{Kumpera A., Spencer D.} (1972) \book{Lie equations.~I: General theory.}
Annals of Math.\ Stud.~\vol{73}. Princeton Univ.\ Press, Princeton, NJ.


\bibitem{Magri}
\by{Magri F.} (1978) A simple model of the integrable equation,
\jour{J.~Math.\ Phys.} \vol{19}:5, 1156--1162.

\bibitem{Manin1979}
\by{Manin Yu.\ I.} (1978)
Algebraic aspects of nonlinear differential equations.
Current problems in mathematics \vol{11}, AN SSSR, VINITI, Moscow, 5--152
(in Russian).

\bibitem{Contractions}
\by{Nesterenko M., Popovych R.} (2006) 
Contractions of low\/-\/dimensional Lie
algebras, \jour{J.~Math.\ Phys.} \vol{47}:12, 123515, 45 pp.

\bibitem{Olver}
\by{Olver P. J.} (1993) \book{Applications of Lie groups to differential
equations}, Grad.\ Texts in Math.\ \vol{107} (2nd ed.), 
Springer\/--\/Verlag, NY. 

\bibitem{Suris}
by{Suris Yu.\ B.} (2003)
\book{The problem of integrable discretization: Hamiltonian approach}.
Progr.\ in Math.\ \vol{219}. Birkh\"auser Verlag, Basel. 

\bibitem{Vaintrob}
\by{Vaintrob A. Yu.} (1997)
Lie algebroids and homological vector fields,
\jour{Russ.\ Math.\ Surv.} \vol{52}:2, 428--429.

\bibitem{Voronov}
\by{Voronov T.} (2002) Graded manifolds and Drinfeld doubles for Lie bialgebroids, in: \book{Quantization, Poisson brackets, and beyond}
(Voronov~T., ed.) Contemp.\ Math.\ \vol{315}, AMS, Providence, RI, 131--168.

\end{thebibliography}
\end{document}